\documentclass[aps,prx,onecolumn,superscriptaddress,notitlepage, 10pt]{revtex4-2}

\usepackage{amsthm}
\usepackage{amsfonts}
\usepackage{amssymb}
\usepackage{mathrsfs}
\usepackage{amsmath}
\usepackage{verbatim}
\usepackage{manfnt}

\usepackage{amsmath}
\usepackage{latexsym}
\usepackage{amsfonts}
\usepackage{amssymb}

\usepackage{lmodern,bm}

\usepackage{bbm,dsfont}
\usepackage{dsfont}
\usepackage{graphicx}
\usepackage{hyperref}

\usepackage{verbatim}

\usepackage[all,cmtip,2cell]{xy}
\UseTwocells

\usepackage{color}
\usepackage[usenames,dvipsnames,svgnames,table]{xcolor}

\newtheorem{proposition?}{Proposition?}

\theoremstyle{plain}
\newtheorem{thm}{Theorem}[section]

\newtheorem{prop}[thm]{Proposition}

\newtheorem*{thm*}{Theorem}

\theoremstyle{definition}

\newcommand{\HH}{\mathcal{H}}
\newcommand{\HS}{\mathcal{H}_{S}}
\newcommand{\BHS}{B(\mathcal{H}_{S})}
\newcommand{\BLX}{B(L^2(X))}
\newcommand{\BBG}{\left( B(\mathcal{H}_{S}) \otimes B(L^{2}(X)) \right)^{G}}
\newcommand{\const}{\boldsymbol{1}}
\usepackage{color}

\newcommand{\LX}{L^2(X)}

\usepackage[normalem]{ulem} 
\newcommand{\his}{\mathcal{H}_{\mathcal{S}}}
\newcommand{\hir}{\mathcal{H}_{\mathcal{R}}}

\theoremstyle{remark}
\newtheorem{remi}[thm]{Remark}

\newcommand{\Y}{\yen}

\newcommand{\ket}[1]{|#1\rangle} %ket

\newcommand{\bra}[1]{\langle#1|} %bra
\newcommand{\tr}[1]{\textrm{tr}\left[#1\right]} %trace
\newcommand{\A}{\mathsf{A}}

\newcommand{\Sy}{\mathcal{S}}

\newcommand{\R}{\mathcal{R}}
\newcommand*\colvec[3][]{
    \begin{pmatrix}\ifx\relax#1\relax\else#1\\\fi#2\\#3\end{pmatrix}
}
\newcommand{\E}{\mathsf{E}}%generic observable
\begin{document}

\title{Quantum Reference Frames on Finite Homogeneous Spaces}

\author{Jan G\l owacki}\email{glowacki@cft.edu.pl}
\affiliation{Center for Theoretical Physics, Polish Academy of Sciences, Al. Lotnik\'ow 32/46, 02-668 Warsaw, Poland}
\author{Leon  Loveridge}
\email{leon.d.loveridge@usn.no }
\affiliation{Quantum Technology Group, Department of Science and Industry Systems, University of South-Eastern Norway, 3616 Kongsberg, Norway}
\author{James Waldron}
\email{james.waldron@newcastle.ac.uk}
\affiliation{
School of Mathematics, Statistics and Physics, Newcastle University, Newcastle upon Tyne, NE1 7RU, United Kingdom}

\begin{abstract}
We present an operationally motivated treatment of quantum reference frames
in the setting that the frame is a covariant positive operator valued measure (POVM) on a finite homogeneous space, generalising the principal homogeneous spaces studied in previous work. We focus on the case that the reference observable is the canonical covariant projection valued measure on the given space, and show that this gives rise to a rank-one covariant POVM on the group, which can be seen as a system of coherent states, thereby making contact with recent work in the perspective-neutral approach to quantum reference frames.
\end{abstract}

\maketitle

\section{Introduction}

In the orthodox treatment of quantum theory, there is an implicit reliance on background (\emph{external}) reference frames which underpin the description of states and observables, and whose interpretation is rooted in classical physics. The topic of \emph{quantum reference frames} (QRFs) seeks to provide a quantum theory in which reference frames are viewed as quantum physical systems and treated \emph{internally}, i.e. as an explicit part of the theoretical description of the quantum physical world. Though the idea can be traced back as far as 1946 \cite{eddington1946fundamental}, the main catalyst for the modern ideas can be found in an exchange regarding the nature of superselection rules in the late 60's \cite{aharonov1967charge,wick1970superselection}, and appears in a more definitive form in the 80s \cite{PhysRevD.30.368}. This subject continues to inspire both pure and applied research, and a variety of formulations now exist; broadly these split into four categories: that motivated by ideas in quantum information (largely summarised in \cite{brs}; see \cite{castro2021relative} for some recent work), that motivated by gauge theory and Dirac quantisation of constrained systems, known as the \emph{perspective-neutral} approach \cite{vanrietvelde2018switching,vanrietvelde2018switching,krumm2020quantum,de2021perspective}, that founded on a direct description of frame transformations \cite{giacomini2019relativistic,de2020quantum}, and that closer to operational ideas in quantum mechanics and quantum measurement theory \cite{loveridge2012quantum,lov6,lov1,lov2,lov4,loveridge2019relative,loveridge2020relational}.

The present contribution follows the latter approach, in which invariant observables are obtained through a \emph{relativisation} 
 map, viewed as the explicit introduction of a quantum reference frame into the description, denoted by $\Y$ in previous works. This map is constructed so as to give rise to positive operator valued measures (POVMs) invariant under a tensor product unitary representation $U_{\Sy}\otimes U_{\mathcal{R}}$ of some locally compact group $G$, given explicitly as 
 \begin{equation}\label{eq:yen1}
     (\Y\circ\A)(\Omega) = \int_G U_{\Sy}(g)\A(\Omega)U_{\Sy}(g)^* \otimes d\E(g).
 \end{equation}
 Here, $\E$---the \emph{frame observable}---is a covariant POVM whose domain is the Borel sets of $G$, i.e.,
 for any $X \in \mathcal{B}(G)$, $U_{\mathcal{R}}(g)\E(X)U_{\mathcal{R}}(g)^* = \E(g.X)$, and $\A$ is a POVM with arbitrary value space containing (measurable) $\Omega$. If $\E$ admits a state which is arbitrarily well localised at the identity in $G$ (which follows from a technical condition known as the norm-1 property \cite{heinonen2003norm}), it can be shown that $\A$ and $\Y\circ\A$ can be made arbitrarily probabilistically close, pointing to an equivalence between the descriptions with internal and externalised frame. Conversely, if $\E$ is bad, there are distributions given by $\A$ which cannot be captured at the invariant level. See \cite{lov1,lov4} for further details of the construction and a more extensive discussion of the physical interpretation.
 %In particular, if $\A = \E$ so that the system is also a frame, the above result shows that the frame is ``good enough" for the reproduction of $\A$ in an invariant manner. 

 In order to account for a wider range of physical scenarios, $G$ as above should be replaced by some homogeneous $G$-space $X$ on which the action of $G$ is (possibly) not free, and $\E$ should therefore be defined on this space. This may arise from e.g. limitations on the availability of a physical frame; in this case the actual frame may not be capable of describing the properties of $\Sy$ fully, or phrased in more relational terms, may not capture the full set of relational properties of $\Sy$ combined with a better frame. For example, $\Sy$ may possess a $G$-covariant PVM built from a collection of rank-$1$ projections, but $\R$ may be ``smaller".
 
 Mathematically, such a move requires a significant update of \eqref{eq:yen1}, and in the present work we provide this generalisation for the case that $G$ and $X$ are finite. If the action of $G$ on $X$ is not free, the presence of a stabiliser/isotropy subgroup necessitates a re-appraisal of the physics of quantum reference frames in contrast to   
 \cite{loveridge2012quantum,lov6,lov1,lov2,lov4,loveridge2019relative,loveridge2020relational}.  Our work differs from that of \cite{de2021perspective}, in which a quantum frame is given by a system of coherent states 
 $S=\{\eta_g ~|~g \in G\}$
 and an \emph{incomplete} frame corresponds to the presence of a non-trivial subgroup of $G$ which stabilises $\eta_{e}$. However, in the setting that the frame observable is the canonical covariant PVM on $X$ acting in $L^2(X)$, we will show that this induces a rank-$1$ POVM on $G$ which can be interpreted as a system of coherent states. We see that the presence of the stabiliser subgroup, through the relativisation and restriction procedures ``enforces" its symmetry back on $\Sy$, in direct analogy to the findings of \cite{de2021perspective}. After briefly discussing the set-up of our framework, we provide the analogue of $\Y$ on a finite homogeneous space, and provide some physical discussion and an illustrative example.

\section{Preliminaries}

A finite set $X$ transforms under the action of a finite group $G$ and quantum reference frames are constructed as a covariant (under a unitary Hilbert space representation of $G$) POVM on $X$. The operational philosophy dictates that physical claims are made at the level of the Born rule probabilities generated by states. Special attention is paid to the canonical PVM on $X$. Two classes of frame states are introduced: the invariant states, which as we will see later are maximally bad frames, and the localised states, which are maximally informative.

\subsection{Setup}
Throughout the rest of this paper, $G$ denotes a finite group with identity $e \in G$, $X$ a finite set with a transitive left $G$-action $\alpha:G \times X \to X$; we write $g \cdot x$ for $\alpha(g,x)$, and $U_{\mathcal{S}}$ is a unitary representation of $G$ in a Hilbert space $\his$ representing some quantum system $\Sy$ (which is arbitrary). We set $\hir := L^2(X)$ as the set of complex functions on $X$. For $x \in X$ we denote by $G_{x} := \{ g \in G \; | \; g \cdot x = x \}$ the stabiliser of $x$.
Note that $G_{g \cdot x} = g G_{x} g^{-1}$.
The map $G/G_{x} \to X$, $gG_{x} \mapsto g \cdot x$ is a $G$-equivariant bijection.
We set $n := |X|$, which is equal to $|G/G_{x}|$ and therefore also to the dimension of $\LX$.

By $U_{R}$ we denote the unitary representation of $G$ on $\LX$ defined by $(g \cdot f)(x) = f(g^{-1} \cdot x)$ for $f \in \LX$ and $g \in G$. Note that if $\delta_{x}$ is the indicator function of $x \in X$ then $g \cdot \delta_{x} = \delta_{g \cdot x}$. The state $\delta_{x}$ may be written $\ket{x}$ and the corresponding rank-$1$ projection $\ket{x}\bra{x}$.
Fixing $x \in X$ and setting $H:=G_x$ identifies $X$ with $G/H$ and $\LX$ with the induced representation $L^{2}(G/H) = \mathrm{Ind}_{H}^{G} \mathbb{C}$. (Since the stabilisers of different $x$ are conjugate, the corresponding induced representations are isomorphic.) Setting $X=G$ with the action given by the group multiplication yields the left regular representation of $G$.

For $x \in X$ we denote by $P_{x} \in \BLX$ the projection onto the subspace spanned by the indicator function $\delta_{x}$ of $x$, and by $P$ the projection valued measure (PVM) defined by 
\[
P : X \to \BLX \; , \; x \mapsto P_{x},
\]
 extended by additivity to subsets of $X$.
The PVM $P$ is equivariant/covariant with respect to the actions of $G$ on $X$ and on $\BLX$, i.e. $P_{g \cdot x} = U_{R}(g) P_{x} U_{R}(g)^* = g \cdot P_{x}$. (The pair $(P,U_{R})$ corresponds to the standard representation of the crossed product $C(X) \rtimes G$ on $L^{2}(X)$ \cite{williams2007crossed}.) We let $G$ act on $\BHS$ and on $\BLX$ by conjugation, so that $g \cdot A = U_{\mathcal{S}}(g) A U_{\mathcal{S}}(g)^*$ for $A \in \BHS$ and $g \in G$, and similarly for $B \in \BLX$; $G$ acts on $\BHS \otimes \BLX$ diagonally. The triple $(U_R,P,\LX)$ is understood as a quantum reference frame, and $P$ the frame observable.

\begin{remi}
We work primarily in the setting that the frame is given as above. However, there is a generalisation in which $P$ is replaced by a covariant POVM $E$, and $\LX$ can be any Hilbert space in which the covariant POVM can be constructed, e.g. it can by smaller than $\LX$.
\end{remi}

\subsection{States and Probabilities}\label{subsec:sap}
Each point $x \in X$ determines a pure state (understood as a rank-$1$ projection) $P_{x}$ corresponding to the function $\delta_{x}$. Pure states of this form are understood as being localised at the corresponding point in $X$; this is motivated through the localisation of the
probability measure to which we will soon turn. We note also that these states are not $G$-invariant unless $G$ acts trivially on $X$. The space of $G$-invariant vectors $\LX^{G}$ in $\LX$ is one dimensional and spanned by the constant function $\const$, since
\begin{align*}
\mathrm{Hom}_{G} ( \mathbb{C} , \mathrm{Ind}_{H}^{G} \mathbb{C} ) \cong \mathrm{Hom}_{H} ( \mathbb{C} , \mathbb{C} ) \cong \mathbb{C} ,
\end{align*}
by Frobenius reciprocity (e.g. \cite{serre1977linear}). The corresponding state $P_{\const}$ is $G$-invariant.

More generally, invariant pure states correspond to vectors that are contained in one-dimensional subrepresentations, because if $\phi \in \LX$ then $U_{\mathcal{S}}(g) P_{\phi} U_{\mathcal{S}}(g)^* = P_{\phi}$ if and only if $U_{\mathcal{S}}(g) \phi = \lambda \phi$ for some $\lambda \in S^{1}$.
Equivalently, for each $x \in X$ and $g \in G$
\begin{equation}
\label{eqn: invariant_pure_state}
\phi(g^{-1} \cdot x) = \lambda \phi(x)
\end{equation}
for some $\lambda \in S^{1}$.
Non-zero vectors of this type correspond to one dimensional subrepresentations of $\LX$.
If $\mathbb{C}_{\chi}$ is such a representation, where $g \cdot v=\chi(g)v$ for $\chi:G \to \mathbb{C}^{\times}$, then
\begin{align*}
\mathrm{Hom}_{G} ( \mathbb{C}_{\chi} , \mathrm{Ind}_{H}^{G} \mathbb{C} ) \cong \mathrm{Hom}_{H} ( \mathbb{C}_{\chi} , \mathbb{C} ) \cong 
\begin{cases}
\mathbb{C} & \mathrm{if } \; \chi(H) = 1 \\
0 & \text{else }
\end{cases}
\end{align*}
so that invariant pure states correspond to one dimensional representations of $G$ that are trivial when restricted to $H$. There always exists an invariant mixed state by the finite dimensionality of $\LX$, given by $\frac{1}{n} \mathbb{I}$ (where $n = |X| = \mathrm{dim}~ L^2(X)$).

Probability measures/distributions are given by the usual trace formula.
For a POVM $E$, 
\begin{equation}
    p^E_{\rho}(x)=\tr{E(x)\rho}
\end{equation}
is a probability distribution on $X$; the corresponding measure is obtained through the additivity of $E$. If $\rho$ is a rank one projection defined by the unit vector $\phi$, we write $p^E_{\phi}(x)$ for the distribution. A state $\rho$ is called localised at $x\in X$ with respect
to $E$ if $p^E_{\rho}(x)=1$, i.e., if upon measurement of $E$, the outcome $x$ is found with probabilistic certainty. In case $X=G$, the POVM $E$ paired with some state gives a probabilistic orientation assignment on $G$.

 POVMs do not share the localisability properties of PVMs. For any PVM $Q$ on $X$, given a projection $Q_x$ there is a state $\rho$ for which $\tr{Q_x\rho} =1$ - simply take (for example) $\rho$ to be the projection onto any unit vector in the range of $Q_x$. For a POVM $E$ on $X$ it may be the case that there is no such state; we will encounter such an example in Subsec. \ref{sec: povm} for a POVM on $G$ corresponding to a system of coherent states. However, we recall that a POVM $E$ satisfies the norm-$1$ property if each $E(x)\neq 0$
has $||E(x)||=1$. The following proposition characterises POVMs with the localisability properties of PVMs, i.e., there is a state for which
$\tr{E(x)\rho}=1$. This unit probability can also be interpreted as $E$ possessing the property associated with $x$ with probabilistic certainty in the state $\rho$.

\begin{prop}\label{prop:no1}
The following are equivalent:
\begin{enumerate}
    \item $E$ satisfies the norm-$1$ property; \label{aa}
    \item For each $E(x) \neq 0$ there is a state $\rho$ such that
    $\tr{E(x)\rho}=1$;\label{bb}
    \item For each $E(x) \neq 0$ there is a pure state $\rho$ such that
    $\tr{E(x)\rho}=1$.\label{cc}
\end{enumerate}
\end{prop}
The proof follows from the spectral radius formula and the finite dimensionality of the Hilbert space.\\

Returning to the canonical PVM $P$, we conclude by observing that the distribution for $P$ and any $G$-invariant pure state is uniform, equal to that given by the state $\frac{1}{n} \mathbb{I}$:
If $P_{\phi}$ is $G$-invariant, then by \eqref{eqn: invariant_pure_state} the function $|\phi(x)|^2$ is constant, with value $\frac{1}{n}$ if $\phi$ is normalised. We have $p^P_{\phi} : X \to \mathbb{R}^{\ge 0}$ is equal to $x \mapsto \mathrm{tr} \left( P_{x} P_{\phi} \right) = \langle \phi , P_{x} \phi \rangle = | \phi(x) |^2 = \frac{1}{n}$, and therefore $p^P_{\phi} = p^P_{\frac{1}{n} \mathbb{I}}$.

% >>>
\section{Relativisation and Restriction}

A quantum system $\Sy$ represented by a complex separable Hilbert space $\his$ is combined with a quantum reference frame $\R$ with Hilbert space $\LX$. Relative observables on $\his \otimes \LX$ are constructed through a generalisation to a finite homogeneous space of
the relativisation ($\Y$) mapping \eqref{eq:yen1}. We provide a family of such maps---one for each $x\in X$---and for these to be well defined the domain must be $G_x$ invariant. As we shall see, the presence of the non-trivial stabiliser means that only $G_x$-invariant observables of $\Sy$ can be ``measured", or more accurately, can be used to compute the statistics of system plus reference. This means that some probability distributions for observables of $\Sy$ cannot be described relative to the given frame, which is therefore understood as \emph{incomplete}, in line with other usage in the literature (e.g. \cite{de2021perspective}). 

\subsection{The $\Y$ map on a homogeneous space}

\begin{prop}\label{prop:props}
For $x \in X$ the map 
\begin{align} 
\nonumber \yen_x : \BHS^{G_{x}} & \to \BBG \\
\label{eqn: yen}
A & \mapsto \sum_{gG_{x} \in G/G_{x}} U_{\mathcal{S}}(g) A U_{\mathcal{S}}(g)^* \otimes P_{g \cdot x}
\end{align}
is a unit preserving, injective C*-algebra homomorphism, where the sum is over a set of representatives of the right cosets of $G_{x}$ in $G$.
\end{prop}

\begin{proof}
First note that the right hand side of \eqref{eqn: yen} well defined because $A$ is $G_{x}$-invariant: if $g,g'$ are both representatives of a coset $gG_x = g'G_x$, then $g = g'h$ for some $h \in G_x$, and
\begin{align*}
U_{\mathcal{S}}(g'h) A U_{\mathcal{S}}(g'h) \circ P_{g'h \cdot x} & = U_{\mathcal{S}}(g')U_{\mathcal{S}}(h) A U_{\mathcal{S}}(h)^* U_{\mathcal{S}}(g)^*  \otimes P_{g'h \cdot x} \\
& = U_{\mathcal{S}}(g') A U_{\mathcal{S}}(g')^* \otimes P_{g' \cdot x}
\end{align*}
where $g'h \cdot x = g' \cdot x$ because $h \in G_{x}$.
If $k \in G$ then 
\begin{align*}
k \cdot \yen_x (A) & = \sum_{gG_{x} \in G/G_{x}} U_{\mathcal{S}}(k) U_{\mathcal{S}}(g) A U_{\mathcal{S}}(g)^* U_{\mathcal{S}}(k)^* \otimes U_{R}(k) P_{g \cdot x} U_{R}(k)^* \\
& = \sum_{gG_{x} \in G/G_{x}} U_{\mathcal{S}}(kg) A U_{\mathcal{S}}(kg)^* \otimes P_{kg \cdot x} \\
& = \sum_{gG_{x} \in G/G_{x}} U_{\mathcal{S}}(g) A U_{\mathcal{S}}(g)^* \otimes P_{g \cdot x}
\end{align*}
which shows that $\yen_x(A)$ is $G$-invariant.
Fix representatives $g_1, \dots, g_n$ of the right cosets of $G_{x}$.
Then the ordered basis $\delta_{g_i \cdot x}$, $1 \le i \le n$ of $L^{2}(X)$ identifies $\BLX$ with $\mathrm{M}_{n}(\mathbb{C})$ and $\BHS \otimes \mathrm{M}_{n}(\mathbb{C})$ with $\mathrm{M}_{n}(\BHS)$.
Under this identification $\yen_x(A)$ is equal to the diagonal matrix
\[
\yen_x (A) = \begin{pmatrix}
A &  &  &  \\
 & U_{\mathcal{S}}(g_1) A U_{\mathcal{S}}(g_1)^* &  \\
 &  & U_{\mathcal{S}}(g_2) A U_{\mathcal{S}}(g_2)^* &  \\
 &  &  & \ddots \\
 &  &  &  & U_{\mathcal{S}}(g_n) A U_{\mathcal{S}}(g_n)^* \\
\end{pmatrix}
\]
from which it follows that the map $\yen_x$ is a unit preserving, injective C*-algebra homomorphism.
\end{proof}
Observables in the image of a $\Y_x$ map are called \emph{relative observables}: they relate system and frame. These correspond in our setting to the relational Dirac observables described in e.g. \cite{de2021perspective}.

\begin{remi}
Replacing $\BLX$ by any $B(\hir)$ for which a covariant POVM $E$ on $X$ may be constructed yields a well defined linear map 
\begin{align*}
\yen^E_x : \BHS^{G_{x}} & \to \left(B(\his) \otimes B(\hir)\right)^G \\
A & \mapsto \sum_{gG_{x} \in G/G_{x}} U_{\mathcal{S}}(g) A U_{\mathcal{S}}(g)^* \otimes E_{g \cdot x};
\end{align*}
however, the injectivity and homomorphism properties are lost, and $E$ may not have rank-$1$ effects. The map is however completely positive.
\end{remi}

\begin{remi}
There are isomorphisms giving $B(\his) \rtimes G \cong (B(\his) \otimes B(L^2(G)))^G \cong B(\his) \otimes B(L^2(G))^G$; under the first isomorphism (given in e.g. \cite{rieffel1980actions}) $\Y_e$ corresponds exactly to the natural inclusion of $B(\his)$ into $B(\his) \rtimes G$. We remark that here $B(L^2(G))^G \cong C^*(G)$ and the isomorphism $B(\his) \rtimes G \cong B(\his) \otimes C^*(G)$ is given in e.g. \cite{Connes} Ch. 2, App. C. The second isomorphism shows that the full algebra of invariants is isomorphic to the original algebra tensored with the invariant algebra of $L^2(G)$, giving a novel  description of the invariant algebra, suggesting that the ``relational information" can be completely factored into one component of the tensor product. This is yet to be fully understood in this context. 

We finally comment that we can replace $G$ with $G/H$ to find $\mathcal{H}(G,H,\mathcal{A},\alpha) \cong (\mathcal{A} \otimes B(L^2(G/H)))^G$; here $\mathcal{H}$ is the ``Hecke crossed product" defined by the data in the parentheses, and $\alpha$ is any $*$-automorphism. $\Y$ factors through the above isomorphism and the injective inclusion of $\mathcal{A}^H$ into $\mathcal{H}(G,H,\mathcal{A},\alpha)$; from here, the properties stated in \ref{prop:props} easily follow. See \cite{Heck!} for more details.
\end{remi}

 \subsection{Restriction map}

 The completely positive conditional expectation $\Gamma_{\omega}$ defined below provides a frame-conditioned observable in $B(\his)$, understood as the description of the system given that the frame has been prepared in the specified state.
 
 For a state $\omega$  on $\BLX$ (we do not make a distinction between the 
positive norm-$1$ linear functional on the algebra and the uniquely associated density operator) we denote by $\Gamma_{\omega}$ the map
\begin{align}
\nonumber \Gamma_{\omega} : \BHS \otimes \BLX & \to \BHS \\
A \otimes B & \mapsto \mathrm{tr}(\omega B) A, \label{eq:res1}
\end{align}
extended linearly to the whole space. This map may be equivalently defined through $\tr{\rho\Gamma_{\omega}(A)}=\tr{\rho \otimes \omega A}$ for all $A \in \BHS \otimes \BLX$ and all states $\rho$ on $\BLX$. Note that $\Gamma_{\omega}$ can be defined on POVMs by composition, and does not preserve PVMs, i.e., for a PVM $Q$, $\Gamma_{\omega}\circ Q$ may not be sharp, as can easily be observed from \eqref{eq:res1}.

\begin{prop}\label{prop:2.2}
The following statements hold.
\begin{enumerate}
\item $( \Gamma_{P_{x}} \circ \yen_{x} ) = \mathrm{id}_{\BHS ^{G_x}}$

\item $( \Gamma_{P_{g\cdot x}} \circ \yen_{x} )$ is equal to the isomorphism
\[
\BHS^{G_{x}} \to \BHS^{G_{g \cdot x}} \; , \; A \mapsto U_{\mathcal{S}}(g) A U_{\mathcal{S}}(g)^* .
\]

\item The following diagram commutes:
\[
\xymatrix{
\BHS^{G_{x}} \ar[rrr]^{\Gamma_{P_{g \cdot x}} \circ \yen_x} \ar[d]_{ U_{\mathcal{S}}(g) (-) U_{\mathcal{S}}(g)^*} &&& \BHS^{G_{g \cdot x}} \ar[d]^{\mathrm{id}} \\
\BHS^{G_{g \cdot x}} \ar[rrr]_{\Gamma_{P_{g \cdot x}} \circ \yen_{g \cdot x} = \mathrm{id}_{\BHS^{G_{g \cdot x}}} } &&& \BHS^{G_{g \cdot x}}
}
\]

\item $\Gamma_{\frac{1}{d} \mathbb{I}} \circ \yen_x$ and $\Gamma_{P_{1}} \circ \yen_x$ are both equal to the surjective map
\begin{equation}
\BHS^{G_x} \to \BHS^{G} \; , \; A \mapsto \frac{1}{|G|} \sum_{g \in G} U_{\mathcal{S}}(g) A U_{\mathcal{S}}(g)^* .
\end{equation}
\end{enumerate}
\end{prop}

\begin{proof}

\begin{enumerate}
\item This is a special case of 2.
\item  Using the fact that the projections $P_{x}, P_{y}$ are orthogonal for $x \ne y$, we have that
\begin{align*}
(\Gamma_{P_{g\cdot x}} \circ \yen_{x}) (A) & = \sum_{kG_{x} \in G/G_{x}} \mathrm{tr}( P_{g\cdot x} P_{k \cdot x} ) U_{\mathcal{S}}(k) A U_{\mathcal{S}}(k)^* \\
& = \mathrm{tr}( P_{g\cdot x}^{2} ) U_{\mathcal{S}}(g) A U_{\mathcal{S}}(g)^* \\
& = U_{\mathcal{S}}(g) A U_{\mathcal{S}}(g)^* .
\end{align*}
The map 
\[
\BHS^{G_{x}} \to \BHS^{g G_{x} g^{-1}} \; ,\; A \mapsto U_{\mathcal{S}}(g) A U_{\mathcal{S}}(g)^*
\]
is a C*-algebra isomorphism, and $gG_{x}g^{-1} = G_{g \cdot x}$.
\item Direct computation.
\item For any invariant state $\rho$, the equivariance of $P$ means, $\tr{P_x \rho}=\tr{P_{g\cdot x}\rho}$ for all $g$ and $x$, and hence the distribution
$y \mapsto p^P_{x}(y)$ is uniform and equal to $\frac{1}{n}$ (recall that $n = {\rm dim}~L^2(X) = |G/G_x|$). Therefore for any invariant $\rho$, 
\begin{align*}
    (\Gamma_{\rho} \circ \yen_{x}) (A) &=\sum_{gG_{x} \in G/G_{x}} U_{\mathcal{S}}(g) A U_{\mathcal{S}}(g)^*  \tr{\rho P_{g.x}}\\
    &= \frac{1}{n}\sum_{gG_{x} \in G/G_{x}} U_{\mathcal{S}}(g) A U_{\mathcal{S}}(g)^*\\
    &= \frac{1}{|G/H|}\frac{1}{|H|}\sum_{gG_{x} \in G/G_{x}} \sum_{h \in G_x} U_{\mathcal{S}}(gh) A U_{\mathcal{S}}(gh)^*,\\
    &= \frac{1}{|G|}\sum_{g \in G} U_{\mathcal{S}}(g) A U_{\mathcal{S}}(g)^*,
\end{align*}
where in the final line we used that $A$ is invariant under $G_x$ and that
$|G|=n |G_x|$.
\end{enumerate}
\end{proof}

We interpret these findings as follows.
Previous work (e.g. \cite{lov1}) considered the setting $X=G$ as a group,  picking out the privileged point $e \in G$, with the analogue of item 1. of \ref{prop:2.2}
reading $(\Gamma_{P_e}\circ \Y) = \mathrm{id}_{B(\his)}$ (where $\Y$ is defined as in \eqref{eq:yen1} for finite $G$). This came with the interpretation that if the frame is localised at the identity (i.e., prepared in the vector state $\delta_e$), there is an exact agreement between 
the description of the system $\Sy$ on its own, and with the inclusion of the appropriately prepared frame, thereby pointing to a condition on the frame that must be satisfied for the usual frameless/external description to apply. The deficiency of that work in contrast to the present formulation is twofold: first, it pertained only to the setting in which the underlying $G$ space is (at least as a set) essentially equal to $G$, and second, it was contingent on an arbitrary choice of base point. 

The present work remedies these issues: the $\Y$ construction is now given as an $X$-parametrised family of maps, allowing for more general scenarios than $X=G$ and restoring the egalitarian understanding of $X$ as a collection of featureless points. This is reflected in item 1. of \ref{prop:2.2}, in which the internal and external descriptions apply to each $\Y_x(A)$, given a state localised at $x \in X$. Moreover, by item 2., for any given $x \in X$, the algebras defined through localisation at $g.x$ are manifestly isomorphic and contain the same physical information---as can be seen explicitly from item 3. We also observe that localisation at $x$ still yields a $G_x$-invariant description, owing to the fact that $\Y_x$ can only make sense on $G_x$-invariant operators. This is a direct analogue of a result in \cite{de2021perspective}, where it is noted that in the presence of a non-trivial stabiliser for the frame, ``the quantum
frame can only resolve those properties of the remaining subsystems that are also invariant under its
isotropy [stabiliser] group". In other words, even if $\Sy$ has ostensibly 
$G_x$ non-invariant properties (more properly understood as manifested with respect to a complete frame), these can never be observed by any frame with
a $G_x$ stabiliser. Item 4. shows that invariant reference states (noting that there always exists at least one pure and one mixed invariant state) gives rise to a finite version of the $G$-twirl, prevalent in the quantum information approach to QRFs \cite{brs}. This is used there to `remove' any external frame dependence, where the frame is understood as localised but the actual value is not known, and the averaging captures the lack of knowledge of the exact orientation. Here, the interpretation is that the quantum reference state is completely indeterminate (with respect to $G$), and thus it is not surprising that the description of $\Sy$ relative to an indeterminate state carries no dependence on $G$. In other words, only invariant states/observables can be used to describe $\Sy$ and the $G$-sensitivity is completely washed out by the $G$-twirl operation.

\subsection{Further analysis}\label{sec: povm}
The canonical PVM $P$ on $X$ defines a rank-$1$ POVM on $G$ by pulling back along the quotient map. This allows us to connect the perspective drawn up in this paper with that based on coherent states given in \cite{de2021perspective}. 

We begin by observing that fixing a base point $x \in X$ determines a map
\begin{equation*}
\label{eqn: orbit}
\pi : G \to X \; , \; g \mapsto g \cdot x .
\end{equation*}
If $H := G_{x}$, then under the isomorphism 
\[
G/H \to X \; , \; gH \mapsto g \cdot x
\]
the map $\pi$ corresponds to the quotient map 
\[
G \to G/H , \; \, g \mapsto gH .
\]
The PVM $P : X \to \BLX$ can be pulled back along $\pi$ to define a map $P \circ \pi : G \to \BLX$. Normalising this map then defines a POVM on $G$:
\[
E: G \to B(L^2(X)) \; , \; g \mapsto \frac{1}{|H|} P_{gH} =: E_{g}.
 \]
The positivity of each $E_g$ is guaranteed by the positivity of  $\frac{1}{|H|}$; this prefactor is introduced to make sure $E(X)= \mathbb{I}$:
\[
\sum_{g \in G} P_{gH} = |H| \sum_{gH \in G/H} P_{gH} .
\]
Again, given a state $\omega$ we call the probability distribution $p^E_{\omega} = \tr{E_{\cdot}\omega}$ localised at $g \in G$ if $p^E_{\omega} = \delta_g$.

\begin{prop}
The following are equivalent:
\begin{enumerate}
    \item $X$ is a principal homogeneous space, i.e., $G/H = G$ \label{a}
    \item $H$ is trivial. \label{b}
    \item $E$ is sharp (projection-valued). \label{c}
    \item $E$ satisfies the norm-$1$ property. \label{d}
    \item For each $g \in G$, there exists a state $\omega$ on $\BLX$ such that $p^E_{\omega}$ is localised at $g$. \label{e}
    \item  For each $g \in G$, there exists a pure state $\omega$ on $\BLX$ such that $p^E_{\omega}$ is localised at $g$. \label{f}
\end{enumerate}
\end{prop}

\begin{proof}
 $1.$ $\Leftrightarrow$ $2.$
Standard.

$2.$ $\Leftrightarrow$ $3.$
$H$ is trivial iff $|H| = 1$ iff $E_{g} = 1/|H| P_{gH}$ is a projection.

$3.$ $\Leftrightarrow$ $4.$ Any sharp observable satisfies the norm-$1$ property. Since $E(g) = 1/|H|P_{gH}$, $E(g)$ can satisfy the norm-$1$ property only if $|H|=1$ in which case $E$ is sharp. The rest follows from Prop. \ref{prop:no1}.
\end{proof}

\begin{remi}
Some of the statements also follow from the observation that if $H$ is non-trivial then $|G| > \mathrm{dim} \; \LX$ and so there cannot exist a PVM from $G$ to $\BLX$.
\end{remi}
\begin{remi}
Since each $E_g$ is rank-1 and $E$ is covariant, we can write $E_g=(1/|H|)\ket{g}\bra{g}_x$, the collection of which therefore make up a coherent state system, making contact with \cite{de2021perspective}. This will be discussed further in subsection \ref{sec: coherent}.
\end{remi}

We can now write the $\Y$ map in terms of $E$; again 
we fix a base point $x \in X$, set $H:= G_{x}$, and identify $X$ with $G/H$.
Using the POVM $E$ one can define a linear map
\begin{align*}
\yen^{E} : \BHS & \to \left( \BHS \otimes \BLX \right)^{G} \\
A & \mapsto \sum_{g \in G} U_{\mathcal{S}}(g) A U_{\mathcal{S}}(g)^{*} \otimes E_{g}
\end{align*}

\begin{prop}
\label{prop: yen_E_factor}
For $A \in \BHS$ it holds that
\begin{align*}
\yen^{E}(A) & = \sum_{gH \in G/H} U_{\mathcal{S}}(g)  \left( \frac{1}{|H|} \sum_{h \in H} U_{\mathcal{S}}(h) A U_{\mathcal{S}}(h)^{*} \right)  U_{\mathcal{S}}(g)^{*} \otimes P_{gh}
\end{align*}
In particular, the map $\yen^{E}$ factors as
\[
\BHS \xrightarrow{\mathrm{av}_{H}} \BHS^{H} \xrightarrow{\yen_x} \left( \BHS \otimes \BLX \right)^{G} 
\]
where $\mathrm{av}_{H}$ is the average over $H$ or $H$-twirl:
\[
\mathrm{av}_{H} : \BHS \to \BHS^{H} \; , \; A \mapsto \frac{1}{|H|} \sum_{h \in H} U_{\mathcal{S}}(h) A U_{\mathcal{S}}(h)^{*} .
\]
\end{prop}

\begin{proof}
We have
\begin{align*}
\sum_{g \in G} U_{\mathcal{S}}(g) A U_{\mathcal{S}}(g)^{*} \otimes E_{g} & = \sum_{gH \in G/H} \sum_{h \in H} U_{\mathcal{S}}(gh) A U_{\mathcal{S}}(gh)^{*} \otimes E_{gh} \\
& = \sum_{gH \in G/H} \sum_{h \in H} U_{\mathcal{S}}(g) U_{\mathcal{S}}(h) A U_{\mathcal{S}}(h)^{*} U_{\mathcal{S}}(g)^{*} \otimes \frac{1}{|H|} P_{ghH} \\
& = \sum_{gH \in G/H} U_{\mathcal{S}}(g) \left( \frac{1}{|H|} \sum_{h \in H} U_{\mathcal{S}}(h) A U_{\mathcal{S}}(h)^{*} \right) U_{\mathcal{S}}(g)^{*} \otimes P_{gH} \\
\end{align*}
as required.
\end{proof}

\begin{remi}
\label{remi: povm_yen}
This shows that one is forced to average observables over $H$.
\end{remi}

Under restriction $\Gamma_{\omega}$, the above result shows clearly how statistics of non-$H$-invariant states/observables for $\Sy$ can never arise; even if $\omega$ is localised on the coset $eH$, the resulting observable is twirled over $H$. 

\subsection{Coherent states}
\label{sec: coherent}
We would like to compare our setup with that of \cite{de2021perspective}, in particular with the notions and results contained in pp10-12.
The POVM $E$ can be understood as corresponding to a system of coherent states. We use Dirac notation to make the comparison as clear as possible.

At the level of \emph{vectors}, we have a map
\[
\tilde{E} : G \to \LX \; , \; g \mapsto \delta_{gH} =: |gH \rangle ,
\]
where the notation $\tilde{E}$ is chosen to agree with \cite{de2021perspective}.
The corresponding map to projections is
\[
G \to \BLX \; , \; g \mapsto |gH \rangle \langle gH| = |H| E_{g} ,
\]
where $E_g=1/|H|P_{gH}$ as defined in Subsec. \ref{sec: povm}. 
The set $\{ |gH\rangle \; | \; g \in G \}$ is a homogeneous space with action
\[
k \cdot |gH \rangle = U_{R}(k) |gH\rangle = |kgH\rangle
\]
In particular, 
\[
|gH\rangle = U_{R}(g) |eH\rangle
\]
and the stabiliser of the base point $|eH\rangle$ is $H$.

\begin{prop}
\label{prop: povm_coherent}
The following are equivalent:
\begin{enumerate}
\item The reference frame $X$ is principal.
\item $H$ is trivial.
\item The POVM $E$ is a PVM.
\item The system of coherent states $\tilde{E}$ is \emph{complete} in the sense of \cite{de2021perspective} p.11, which by definition means that the stabiliser of each $|gH\rangle$ is trivial.
\item The system of coherent states $\tilde{E}$ is \emph{ideal} in the sense of \cite{de2021perspective}, p.11, which by definition means that the vectors $|gH\rangle$ are orthogonal to each other.
\end{enumerate}
If any of the above conditions do not hold, then the system of coherent states $\tilde{E}$ is \emph{incomplete} in the sense of \cite{de2021perspective}, p.11.
\end{prop}

The resolution of the identity property, up to scaling, holds:
\[
\sum_{g \in G} |gH \rangle \langle gH| = |H| \; \mathbb{I} ,
\]
where the factor of $|H|$ arises because $E_{g} = \frac{1}{|H|} |gH \rangle \langle gH|$. This corresponds to \cite{de2021perspective}, p.11 (5). 
%Here, this property holds for the `seed vector' $|eH\rangle$, but one can check directly that this vector satisfies the requirements in \cite{de2021perspective}. p.12 Ex.1 by decomposing the induced representation $L^{2}(G/H)$.

The POVM $E$ can be recovered from $\tilde{E}$ as in \cite{de2021perspective}, p. 11 (6):
\[
E_{g} = \frac{1}{|H|} |\tilde{E}_{g}\rangle \langle \tilde{E}_{g}| = |gH\rangle \langle gH| .
\]
The normalisation $|H|$ is equal to $\mathrm{Vol}(H)$ in \cite{de2021perspective}, p.12 (12).
The factorisation of $\yen^{E}$, Proposition \ref{prop: yen_E_factor} above, is related to the factorisation of the integral in \cite{de2021perspective} p.12 (12).

\subsection{Example: Permutations}

 For $n > 2$ the symmetric group $S_{n}$ of permutations of the set 
 $X = \{x_1,\dots,x_n\}$ provides an example of a homogeneous space 
 which is not principal. For the action we write 
\[
\sigma \cdot x_{i} = x_{\sigma(i)}
\]
for $\sigma \in S_n$ and $1 \le i \le n$.
The stabiliser of each element of $X$ is a symmetric group on $n-1$ symbols.
Choosing $x_n$ for clarity, the stabiliser is naturally identified with the symmetric group $S_{n-1} \subset S_{n}$, and there is an isomorphism of $S_{n}$-sets
\[
S_{n} / S_{n-1} \xrightarrow{\cong} X \; , \; \sigma S_{n-1} \mapsto \sigma \cdot x_{n} = x_{\sigma(n)} .
\]
As an $S_n$ representation the Hilbert space $\LX$ decomposes as 
\[
\LX \cong \mathrm{Ind}_{S_{n-1}}^{S_n} \mathbb{C} = \mathbb{C} \oplus \mathrm{Stand}
\]
where $\mathrm{Stand}$ is the standard $n-1$ dimensional irreducible representation of $S_{n}$, which is
\[
\mathrm{Stand} := \left\{ f \in L^{2}(X) \; | \; \sum_{i=1}^{n} f(x_{i}) = 0 \right\} .
\]
In particular, $P_{\const}$ is unique $S_n$-invariant pure state on $\BLX$.

Let $\HH$ be a fixed Hilbert space, $\HS = \HH^{\otimes n}$, and $S_{n}$ act on $\HS$ by permuting tensor factors.
The corresponding action on $B(\HH^{\otimes n}) \cong B(\HH)^{\otimes n}$ is by permuting tensor factors.

The domain of $\yen_{j}$ is the algebra
\[
( B(\HH)^{\otimes n} )^{(jn)S_{n-1}(jn)}
\]
where $(jn) S_{n-1} (jn)$ is the stabiliser of $x_{j}$, and $S_{n-1}$ is the stabiliser of $x_n$.
In particular, the domain of $\yen_{n}$ is the algebra
\[
B(\HH^{\otimes n}) ^{S_{n-1}} \cong \left( B(\HH)^{\otimes n} \right)^{S_{n-1}} = \left( B(\HH)^{\otimes n-1} \right)^{S_{n-1}} \otimes B(\HH) .
\]
If 
\[
A = A_{1} \otimes A_{2} \otimes \cdots \otimes A_{n} \in B(\HH)^{\otimes n} \cong B(\HH^{\otimes n})
\] 
is an element of $\BHS^{S_{n-1}}$ then 
\begin{align*}
\yen_{n} (A) & = \sum_{\sigma \in S_{n-1}} A_{ \sigma^{i}(1) } \otimes \cdots \otimes A_{ \sigma^{i}(n) } \otimes P_{\sigma^{i}(1)} .
\end{align*}
For the maps $\Gamma_{\omega}$ defined by localised states, each of the maps 
\[
\Gamma_{P_{i}} \circ \yen_{j} : ( B(\HH)^{\otimes n} ) ^{(jn)S_{n-1}(jn)} \to ( B(\HH)^{\otimes n} )^{(in)S_{n-1}(in)}
\]
is an isomorphism, where $1 \le i,j \le n$.
In particular, for $j=n$, these maps are 
\[
(\Gamma_{P_{i}} \circ \yen_{n}) (A_1 \otimes \cdots \otimes A_{n}) = A_{\sigma^{i}(1)} \otimes \cdots \otimes A_{\sigma^{i}(n)} .
\]
This exemplifies the behaviour witnessed in the general setting: the restricted relativised observables can be perfectly recovered up to $S_{n-1}$ invariance, but no better.

\section{Concluding Remarks}
We have given a construction of quantum reference frames suited to finite homogeneous spaces and analysed the situation in which the reference observable is given by the canonical covariant PVM, which, as we have seen, reduces to the coherent state set-up given in \cite{de2021perspective}. The case of an unsharp reference observable has not been analysed here, but it is clear that it gives something more general than a system of coherent states on $G$, and in this sense the setting we have provided is more general. It is not clear if other advantages are afforded from the homogeneous space point of view. On the one hand it appears more classical than the prescription of \cite{de2021perspective}, since it appears to be predicated upon the existence of some underlying space $X$. On the other, $X$ can be interpreted as the value space of some observable $E$, and as such is more operational. Future work, involving also the compact and locally compact settings, should shed further light on the relative merits of each approach.  

\section*{Acknowledgements} LL would like to thank the Theoretical Visiting Sciences Programme (TSVP) at the Okinawa Institute of Science and Technology for enabling his visit, and for the generous hospitality and excellent working conditions during his time there, during which part of this work was completed. Thanks are also due to M. H. Mohammady for a careful reading of an earlier draft of this manuscript.

\bibliographystyle{apsrev4-2}
\bibliography{bib}

\end{document}